\documentclass[a4paper, runningheads,14pt]{llncs}

\usepackage[utf8]{inputenc}
\usepackage[T1]{fontenc}
\usepackage[english]{babel}
\usepackage{amsfonts, amsmath, amssymb}
\usepackage[bb=boondox, cal=cm]{mathalfa}
\usepackage{old-arrows}
\usepackage{array}
\usepackage{subcaption}
\usepackage{color}
\usepackage{pgfplots}
\pgfplotsset{compat=1.16}
\usepackage{stmaryrd}
\usepackage{float}
\usepackage{enumitem}
\usepackage{pdfpages}
\usepackage{tcolorbox}
\usepackage{comment}
\usepackage{graphicx}
\usepackage{graphics}
\usepackage{xspace}
\usepackage{cite}

\newcommand{\cc}[1]{\mathcal{#1}}  % calligraphic
\newcommand{\cb}[1]{\mathbb{#1}}  % doubled (IN, IR, ...)
\newcommand{\csf}[1]{\normalfont{\textsf{#1}}}  % sans serif
\newcommand{\csmc}[1]{\textsc{#1}}  % small caps
\renewcommand{\max}{\csf{max}_{\subseteq}} % max
\newcommand{\card}[1]{\vert #1 \vert}  % size of a set
\DeclareMathOperator{\imp}{\rightarrow}  % implication arrow
\DeclareMathOperator{\detm}{\uparrow}  % up-perspectivity
\DeclareMathOperator{\jdet}{\downarrow}  % down-perspectivity
\DeclareMathOperator{\J}{\cc{J}}  % join-irreducible set
\DeclareMathOperator{\M}{\cc{M}}  % meet-irreducible set
\DeclareMathOperator{\Hp}{\cc{H}}  % hypergraph 
\DeclareMathOperator{\E}{\cc{E}}  % (hyper)edges
\DeclareMathOperator{\mis}{\csf{MIS}}  % lattice
\DeclareMathOperator{\is}{\csf{IS}}  % independent sets
\DeclareMathOperator{\NP}{\csf{NP}}  % NP
\DeclareMathOperator{\Po}{\csf{P}}  % P
\DeclareMathOperator{\C}{c}  % Caratheodory number
\DeclareMathOperator{\poly}{\csf{poly}}  % Caratheodory number

\DeclareMathOperator{\K}{\cc{K}}  % Keys
\newcommand{\U}{X}  % Universe for variables
\newcommand{\graph}{G_c}  % consistency-graph
\newcommand{\edges}{E_c}  % edges of graph
\newcommand{\IS}{\Sigma}  % implication system
\newcommand{\cl}{\phi}  % closure operator
\newcommand{\cs}{\cc{F}}  % closure system
\newcommand{\Sol}{\csf{maxCC}(\IS, \graph)}  % Solutions

\newcommand{\ie}{i.e.,\xspace}

\newcommand{\st}{{\emph{s.t.}\xspace}}

\newenvironment{decproblem}[3]{
\vspace{0.5em}
\begin{itemize}
	\item[] \csmc{#1}
	\item[] \textbf{Input:} #2
	\item[] \textbf{Output:} #3
\end{itemize}
}{\vspace{0.5em}}

\definecolor{midnight}{RGB}{44, 62, 80}
\definecolor{belize}{RGB}{0, 103, 176}
\definecolor{teal}{RGB}{0, 150, 136}
\definecolor{amethyst}{RGB}{155, 89, 182}
\definecolor{asbestos}{RGB}{127, 140, 141}
\definecolor{clouds}{RGB}{236, 240, 241}
\definecolor{grass}{HTML}{97CE68}
\definecolor{alizarine}{RGB}{231, 76, 60}

\usepackage[colorlinks=true, citecolor=belize]{hyperref}

\begin{document}

\title{Enumerating maximal consistent closed sets in closure systems}

\titlerunning{Enumerating maximal consistent closed sets}

%%
%%\titlerunning{Abbreviated paper title}
%% If the paper title is too long for the running head, you can set
%% an abbreviated paper title here
%%
\author{Lhouari Nourine \and Simon Vilmin \thanks{The second author is funded by the 
CNRS, France, ProFan project.}}
\authorrunning{L. Nourine \and S. Vilmin}

%% First names are abbreviated in the running head.
%% If there are more than two authors, 'et al.' is used.
%%
\institute{LIMOS, Université Clermont Auvergne, Aubière, France \\
\email{lhouari.nourine@uca.fr},
\email{simon.vilmin@ext.uca.fr}}

\maketitle 

\begin{abstract}
Given an implicational base, a well-known representation for a closure system, an 
inconsistency binary relation over a finite set, we are interested in 
the problem of 
enumerating all maximal consistent closed sets (denoted by \csmc{MCCEnum} for short). 
We show that \csmc{MCCEnum} cannot be solved in output-polynomial time unless $\Po = 
\NP$, even for lower bounded lattices.
We give an incremental-polynomial time algorithm to solve \csmc{MCCEnum} for closure 
systems with constant Carath\'{e}odory number.
Finally we prove that in biatomic atomistic closure systems \csmc{MCCEnum} can be solved 
in output-quasipolynomial time if minimal generators obey an independence condition, 
which holds in atomistic modular lattices. 
For closure systems closed under union (\ie distributive), \csmc{MCCEnum} is solved by a 
polynomial delay algorithm \cite{hirai2018compact, kavvadias2000generating}.

\end{abstract}

\section{Introduction} 
\label{sec:introduction}

In this paper, we consider binary inconsistency relations (\ie graphs) over 
implicational bases, a well-known representation for closure systems \cite{wild2017joy, 
bertet2018lattices}.
More precisely, we seek to enumerate maximal closed sets of a closure system given by an 
implicational base that are consistent with respect to an inconsistency relation.
We call this problem \csmc{Maximal Consistent Closed Sets Enumeration}, or 
\csmc{MCCEnum} for short.

This problem finds applications for instance in minimization of sub-modular 
functions \cite{hirai2018compact} or argumentation frameworks 
\cite{dung1995acceptability}.
It is moreover a particular case of dualization in closure systems given by an 
implicational bases, ubiquitous in computer science \cite{
demetrovics1992functional, fredman1996complexity, bertet2018lattices}.
This latter problem however cannot be solved in output-polynomial time unless $\Po = \NP$ 
\cite{babin2017dualization} even when the input implicational base has premises of 
size at most two \cite{defrain2020dualization}.
When restricted to graphs and implicational bases with premises of size one, or posets 
equivalently, the problem can be solved in polynomial delay 
\cite{kavvadias2000generating, hirai2018compact}.

More generally, inconsistency relations combined with posets appear also in event 
structures \cite{nielsen1981petri}, representations of median-semilattices 
\cite{barthelemy1993median} or cubical complexes \cite{ardila2012geodesics} in 
which the term \textit{``inconsistency''} is used.
Recently in \cite{hirai2018compact, hirai2020compact}, the authors derive a 
representation for modular semi-lattices based on inconsistency and projective ordered 
spaces \cite{herrmann1994geometric}.
Furthermore, they characterize the cases where given an implicational base and an
inconsistency relation, maximal consistent closed sets coincide with maximal independent 
sets of the inconsistency relation, seen as a graph.

In our contribution, we show first that enumerating maximal consistent closed sets 
cannot be solved in output-polynomial time unless $\Po = \NP$, a surprising result which 
further emphasizes the hardness of dualization in lattices given by implicational bases 
\cite{babin2017dualization, defrain2020dualization}.
In fact, we show that this problem is already intractable for the well-known 
class of lower bounded lattices \cite{freese1995free, day1970simple, 
adaricheva2017optimum}.
On the positive side, we show that when the maximal size of minimal generators is bounded 
by a constant, the problem can be solved in incremental-polynomial time.
As a direct corollary, we obtain that \csmc{MCCEnum} can be solved efficiently in 
a several classes of convex geometries where this parameter, also known as the 
Carath\'{e}odory number, is constant 
\cite{korte2012greedoids}.
Finally, we focus on biatomic atomistic closure systems \cite{bennett1987biatomic, 
birkhoff1985convexity}.
We show that under an independence condition, the size of a minimal generator is 
logarithmic in the size of the groundset.
As a consequence, we get a quasi-polynomial time algorithm for enumerating maximal 
consistent closed sets which can be applied to the well-known class of atomistic modular 
lattices \cite{herrmann1994geometric, wild2000optimal, stern1999semimodular, 
gratzer2011lattice}.

The rest of the paper is organized as follows.
Section \ref{sec:preliminaries} gives necessary definitions about closure 
systems and implicational bases.
In Section \ref{sec:hardness} we show that \csmc{MCCEnum} cannot be solved in 
output-polynomial time, in particular for lower bounded closure systems.
In Section \ref{sec:mingen}, we show that if the size of a minimal generator is bounded 
by a constant, \csmc{MCCEnum} can be solved efficiently.
Section \ref{sec:biatomic} is devoted to the class of biatomic atomistic closure systems. 
We conclude with open questions and problems in \ref{sec:conclusion}.

\section{Preliminaries} 
\label{sec:preliminaries}

All the objects considered in this paper are finite. 
Let $\U$ be a set.
We denote by $2^{\U}$ its powerset.
For any $n \in \cb{N}$, we write $[n]$ for the set $\{1, \dots, n\}$.
We will sometimes use the notation $x_1 \dots x_n$ as a shortcut for $\{x_1, 
\dots ,x_n\}$.
The size of a subset $A$ of $\U$ is denoted by $\card{A}$.
If $\Hp = (\U, \E)$ is a hypergraph, we denote by $\is(\Hp)$ its independent sets (or 
stable sets).
We write $\mis(\Hp)$ for its maximal independent sets.
Similarly, if $G = (\U, E)$ is a graph, its independent sets (resp. maximal independent 
sets) are written $\is(G)$ (resp. $\mis(G)$).

We recall principal notions on lattices and closure systems \cite{gratzer2011lattice}.
A mapping $\cl \colon 2^{\U} \to 2^{\U}$ is a \emph{closure operator} if for any $Y, Z 
\subseteq \U$, $Y \subseteq \cl(Y)$ (extensive), $Y \subseteq Z$ implies $\cl(Y) 
\subseteq \cl(Z)$ (isotone), and $\cl(\cl(Y)) = \cl(Y)$ (idempotent).
We call $\cl(Y)$ the \emph{closure} of $Y$.
The family $\cs = \{\cl(Y) \mid Y \subseteq \U\}$ ordered by set-inclusion forms a 
\emph{closure system} or \emph{lattice}.
A closure system $\cs \subseteq 2^{\U}$ is a set system such that
$\U \in \cs$ and for any $F_1, F_2 \in \cs$, $F_1 \cap F_2$ also belongs to $\cs$.
Elements of $\cs$ are \emph{closed sets} and we say that $F$ is \emph{closed} if $F \in 
\cs$.
Each closure system $\cs$ induces a unique closure operator $\cl$ such that $\cl(Y) = 
\bigcap \{F \in \cs \mid Y \subseteq F\}$, for any $Y \subseteq \U$.
Thus, there is a one-to-one correspondence between closure systems and operators.
Without loss of generality, we will assume that $\cl$ and $\cs$ are \emph{standard}: 
$\cl(\emptyset) = \emptyset$ and for any $x \in \U$, $\cl(x) \setminus \{x\}$ is closed.
Note that $\emptyset$ is thus the minimum element of $\cs$, called the \emph{bottom}.
Similarly, $\U$ is the \emph{top} of $\cs$.

Let $\cl$ be a closure operator with corresponding closure system $\cs$.
Let $F_1, F_2 \in \cs$.
We say that $F_1$ and $F_2$ are \emph{comparable} if $F_1 \subseteq F_2$ or $F_2 
\subseteq F_1$.
They are \emph{incomparable} otherwise.
A subset $\cc{S}$ of $\cs$ is an \emph{antichain} if its elements are pairwise 
incomparable.
If for any $F \in \cs$, $F_1 \subset F \subseteq F_2$ implies $F = F_2$, we say that 
$F_2$ \emph{covers} $F_1$, and denote it $F_1 \prec F_2$.
An \emph{atom} is a closed set covering the bottom $\emptyset$ of $\cs$.
Dually, a \emph{co-atom} is a closed set covered by the top $\U$ of $\cs$.
We denote by $\cc{C}(\cs)$ the set of co-atoms of $\cs$.
Let $M \in \cs$.
We say that $M$ is \emph{meet-irreducible} in $\cs$ if for any $F_1, F_2 \in \cs$, $M = 
F_1 \cap F_2$ entails either $F_1 = M$ or $F_2 = M$.
In this case, $M$ has a unique cover $M^*$ in $\cs$.
The set of meet-irreducible elements of $\cs$ is denoted by $\M(\cs)$.
Dually, $J \in \cs$ is a \emph{join-irreducible} element of $\cs$ if for any $F_1, F_2 
\in \cs$, $J = \cl(F_1 \cup F_2)$ implies $J = F_1$ or $J = F_2$.
Then, $J$ covers a unique element $J_*$ in $\cs$.
We denote by $\J(\cs)$ the join-irreducible elements of $\cs$.
When $\cs$ and $\cl$ are standard, there is a one-to-one correspondence between $\U$ and 
$\J(\cs)$ given by $\J(\cs) = \{\cl(x) \mid x \in \U\}$.
Furthermore, $x_* = \cl(x)_* = \cl(x) \setminus x$.
Consequently, we will identify $\U$ with $\J(\cs)$.

Let $x \in \U$.
A \emph{minimal generator} of $x$ is an inclusion-wise minimal subset $A_x$ of $\U$ such 
that $x \in \cl(A_x)$.
%We denote by $\gen(x)$ the set of minimal generators of $x$, \ie $\gen(x) = 
%\min\{A \subseteq \U \mid x \in \cl(A)\}$.
We consider $\{x\}$ as a trivial minimal generator of $x$.
Following \cite{korte2012greedoids}, the \emph{Carath\'{e}odory number} $\C(\cs)$ of 
$\cs$ is the least integer $k$ such that for any $A \subseteq \U$ and any $x 
\in \U$, $x \in \cl(A)$ implies the existence of some $A' \subseteq A$ with $\card{A'} 
\leq k$ such that $x \in \cl(A')$.
At first, this notion was used for convex geometries, but 
its definition applies to any closure system.
Moreover, the Carath\'{e}odory number of $\cs$ is the maximal possible size of a 
minimal generator (see Proposition 4.1 in \cite{korte2012greedoids}, which can be applied 
to any closure system).
A \emph{key} of $\cs$ is a minimal subset $K \subseteq \U$ such that $\cl(K) = \U$.
We denote by $\K$ the set of keys of $\cs$.
The number of keys in $\K$ is denoted by $\card{\K}$.
It is well-known (see for instance \cite{demetrovics1992functional}) that maximal 
independent 
sets $\mis(\K)$ of $\K$, viewed as a hypergraph over $\U$, are exactly co-atoms of $\cs$.
We define arrow relations from \cite{ganter2012formal}.
Let $x \in \U$ and $M \in \M(\cs)$.
We write $x \detm M$ if $x \notin M$ but $x \in M^*$.
Dually, we write $M \jdet x$ if $x \notin M$ but $x_* \subseteq M$.

We move to implicational bases \cite{bertet2018lattices, wild2017joy}. 
An \emph{implication} si an expression of the form $A \imp B$ with $A, B \subseteq \U$.
We call $A$ the \emph{premise} and $B$ the \emph{conclusion}.
A set $\IS$ of implications over $\U$ is an \emph{implicational base} over $\U$.
We denote by $\card{\IS}$ the number of implications in $\IS$.
A subset $F \subseteq \U$ \emph{satisfies} or \emph{models} $\IS$ if for any $A \imp B 
\in \IS$, $A 
\subseteq F$ implies $B \subseteq F$.
The family $\cs = \{F \subseteq \U \mid F \text{ satisfies } \IS 
\}$ is a closure system whose induced closure operator $\phi$ is the \emph{forward 
chaining algorithm}.
This procedure starts from any subset $Y$ of $\U$ and constructs a sequence $Y = Y_0 
\subseteq \dots \subseteq Y_k = \phi(Y)$ of subsets of 
$\U$ such that for any $i \in [k]$, $Y_i = Y_{i - 1} \cup \{B \mid \exists A \imp B \in 
\IS \, \st \, A \subseteq Y_{i - 1} \}$.
The algorithm stops when $Y_{i - 1} = Y_i$.

\begin{figure}[ht!]
\centering 
\begin{subfigure}{0.4\textwidth}
	\centering
	\includegraphics[scale=1.0]{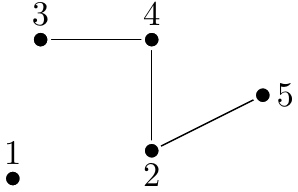}
\end{subfigure}
\begin{subfigure}{0.5\textwidth}
	\centering
	\includegraphics[scale=0.85]{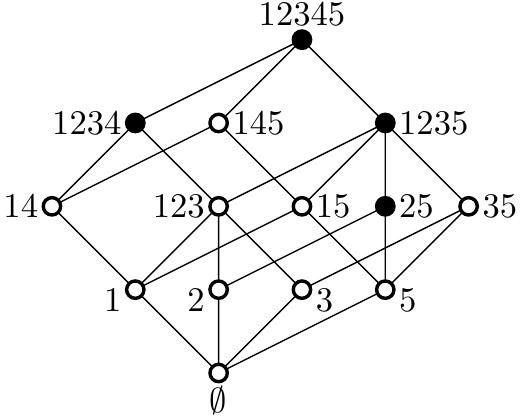}
\end{subfigure}
\caption{
On the left, a consistency-graph $\graph$ over $\U = \{1, 2, 3, 4, 5\}$ with inconsistent 
pairs $34$, $24$ and $25$. 
On the right, the closure system associated to $\IS = \{13 \imp 2, 12 \imp 3, 23 
\imp 1, 4 \imp 1\}$. 
Black and white dots stand for inconsistent and consistent closed sets respectively.
We have $\Sol = \{145, 123, 35\}$.}
\label{fig:running-ex}
\end{figure}

We now introduce our main problem.
Following \cite{hirai2018compact, hirai2020compact, ardila2012geodesics} we call an 
\emph{inconsistency relation} any symmetric and irreflexive relation over $\U$.
Such a relation is sometimes called a \emph{site} \cite{barthelemy1993median} or a 
\emph{conflict relation} \cite{nielsen1981petri}.
Usually, inconsistency relations need to satisfy more conditions in order to capture 
median or modular-semilattices \cite{barthelemy1993median, hirai2020compact}.
As we do not need further restrictions here, we can choose to model inconsistency as a 
graph $\graph = (\U, \edges)$, and call it a \emph{consistency-graph}.
An edge $uv$ of $\edges$ represents an \emph{inconsistent pair} of elements in $\U$.
A subset $Y$ which does not contain any inconsistent pair (\ie an independent set of 
$\graph$) is called \emph{consistent}.
Let $\IS$ be an implicational base over $\U$ and a $\graph = (\U, 
\edges)$ consistency-graph.
We denote by $\Sol$ the set of maximal consistent closed sets of $\cs$, that 
is $\Sol = \max(\cs \cap \csf{IS}(\graph))$.
An example of implicational base along with a consistency-graph is given in Figure 
\ref{fig:running-ex}. 
Our problem is the following.

\begin{decproblem}
{Maximal Consistent closed-sets Enumeration (MCCEnum)}
{An implicational base $\IS$ over $\U$, a non-empty consistency-graph $\graph = (\U, 
\edges)$. }
{The set $\Sol$ of maximal consistent closed sets of $\cs$ with respect 
	to $\graph$.}
\end{decproblem}

Remark that $\U$ is part of the input.
If $\graph$ is empty, \csmc{MCCEnum} is easy to solve as $\U$ is the unique element of 
$\Sol$. 
Hence, we will assume without loss of generality that $\graph$ is not empty.
If $\IS$ is empty, then \csmc{MCCEnum} is equivalent to the enumeration of maximal 
independent sets of a graph which can be efficiently solved \cite{johnson1988generating}.
If premises of $\IS$ have size $1$, the problem also reduces to maximal independent sets 
enumeration \cite{kavvadias2000generating, hirai2018compact}.
In \cite{hirai2020compact} the authors identify, for a fixed $\IS$, the 
consistency-graphs $\graph$ such that $\mis(\graph) = \Sol$.

We conclude with a recall on enumeration algorithms \cite{johnson1988generating}.
Let $\cc{A}$ be an algorithm with input $x$ and output a set of solutions $R(x)$.
We denote by $\card{R(x)}$ the number of solutions in $R(x)$.
We assume that each solution in $R(x)$ has size $\poly(\card{x})$.
The algorithm $\cc{A}$ is running in \emph{output-polynomial} time if its execution time 
is bounded by $\poly(\card{x} + \card{R(x)})$.
It is \emph{incremental-polynomial} if for any $1 \leq i \leq \card{R(x)}$, the time 
spent between the $i$-th and $i + 1$-th output is bounded by $\poly(\card{x} + i)$, and 
the algorithm stops in time $\poly(\card{x})$ after the last output.
If the delay between two solutions output and after the last one is $\poly(\card{x})$, 
$\cc{A}$ has \emph{polynomial-delay}. 
Note that if $\cc{A}$ is running in incremental-polynomial time, it is also 
output-polynomial.
Finally, we say that $\cc{A}$ runs in \emph{output-quasipolynomial} time if is execution 
time is bounded by $N^{\csf{polylog}(N)}$ where $N = \card{x} + \card{R(x)}$.

\section{Closure systems given by implicational bases} 
\label{sec:hardness}

We show that \csmc{MCCEnum} cannot be solved in output-polynomial time unless $\Po = \NP$.
To do so, we use a reduction from the problem of enumerating 
co-atoms of a closure system.

\begin{decproblem}
	{Co-atoms Enumeration (CE)}
	{An implicational base $\IS_Y$ over $Y$.}
	{The co-atoms $\cc{C}(\cs_Y)$ of the closure system $\cs_Y$ associated to $\IS_Y$.}
\end{decproblem}

It is proved by Kavvadias et al. in \cite{kavvadias2000generating} that \csmc{CE} admits 
no output-polynomial time algorithm unless $\Po = \NP$.
Our first step is to prove the following lemma.

\begin{lemma} \label{lem:reduction}
Let $\IS_Y$ be an implicational base over $Y$.
Let $\U = Y \cup \{u, v \}$, $\IS = \IS_Y \cup \{Y \imp uv\}$ and let $\graph = (\U, 
\edges = \{uv\})$ be a consistency-graph.
The following equality holds:
	
\begin{equation} \label{eq:equality}
	\Sol = \bigcup_{C \in \cc{C}(\cs_Y)} \{C \cup \{u\}, C \cup \{v\}\}
\end{equation}

\end{lemma}

\begin{proof}
Let $C \in \cc{C}(\cs_Y)$.
We show that $C \cup \{u\}$ and $C \cup \{v\}$ are in $\Sol$.
As no implication of $\IS$ has $u$ or $v$ in its premise, we have that $C \cup \{u\}$ and 
$C \cup \{v\}$ are consistent and closed with respect to $\IS$.
Let $y \in Y \setminus C$.
As $C$ is a co-atom of $\cs_Y$, it must be that $\cl_Y(C \cup \{y\}) = Y$.
As $Y \imp uv$ is an implication of $\IS$, it follows that $uv \subseteq \cl(C \cup \{u, 
y\})$.
Thus, for any $x \in \U \setminus (C \cup\{u\})$, $\cl(C \cup \{u, x\})$ is inconsistent.
We conclude that $C \cup \{u\} \in \Sol$.
Similarly we obtain $C \cup \{v\} \in \Sol$.

Let $S \in \Sol$.
We show that $S$ can be written as $C \cup \{u\}$ or $C \cup \{v\}$ for some co-atom $C$ 
of $\cs_Y$.
First, let $F$ be a consistent closed set in $\cs$ such that $u \notin F$ and $v 
\notin F$.
As $\IS$ has no implication with $u$ or $v$ in its premise, it follows that both $F \cup 
\{u\}$ and $F \cup \{v\}$ are closed and consistent.
Hence, either $u \in S$ or $v \in S$.
Without loss of generality, let us assume $u \in S$.
Let $C = S \setminus \{u\}$.
As $S \in \Sol$, it is closed with respect to $\IS_Y$ and does not contain $Y$.
Thus, $C \in \cs_Y$ and $C \subset Y$.
Let $y \in Y \setminus C$.
As $S \in \Sol$, it must be that $\cl(S \cup \{y\})$ contains the inconsistent pair $uv$ 
of $\graph$.
Hence, $Y \subseteq \cl(S \cup \{y\})$ by construction of $\IS$.
Consequently, we have that $Y = \cl_Y(C \cup \{y\})$ for any $y \in Y \setminus 
C$.
Hence, we conclude that $C \in \cc{C}(\cs_Y)$ as expected.
\qed
\end{proof}

Therefore, if there is an algorithm solving \csmc{MCCEnum} in output-polynomial time, 
it can be used to solve \csmc{CE} within the same running time using the reduction of 
Lemma \ref{lem:reduction}.
Consequently, we obtain the following theorem.

\begin{theorem} \label{thm:coNP-complete}
The problem \csmc{MCCEnum} cannot be solved in output-polynomial time unless $\Po = \NP$.
\end{theorem}

In fact, we can strengthen the preceding theorem by a careful analysis of the 
closure system used in the reduction in \cite{kavvadias2000generating}.
More precisely, we show that the problem remains untractable for lower bounded 
closure systems.
These have been introduced with the doubling construction in 
\cite{day1970simple} and then studied in 
\cite{bertet2002doubling, freese1995free, adaricheva2017optimum}.
A characterization of lower bounded lattices is given in 
\cite{freese1995free} in terms of the $D$-relation.
This relations relies on $\J(\cs)$ and we say that $x$ depends on $y$, denoted by $xDy$ 
(recall that we identified $\U$ with $\J(\cs)$) if there exists a meet-irreducible 
element $M \in \M(\cs)$ such that $x \detm M \jdet y$.
A $D$-cycle is a sequence $x_1, \dots, x_k \subseteq \U$ such that $x_1 D x_2 D \dots D 
x_k D x_1$.

\begin{theorem}(Reformulated from Corollary 2.39, \cite{freese1995free}) 
\label{cor:freese-LB} A closure system $\cs$ is lower bounded if and only if it contains 
no $D$-cycle.
\end{theorem}

\begin{corollary}
The problem \csmc{MCCEnum} cannot be solved in output-polynomial time unless $\Po = \NP$, 
even in lower bounded closure systems.
\end{corollary}

\begin{proof}
Consider a positive 3-CNF over $n$ variables and $m$ clauses
\[ \psi(x_1, \dots, x_n) = \bigwedge_{i = 1}^{m} C_i = 
\bigwedge_{i = 1}^{m} (x_{i, 1} \lor x_{i, 2} \lor x_{i, 3}) \]
Let $Y = \{x_1, \dots x_n, y_1, \dots, y_m, z\}$ and consider the following sets of 
implications:
\begin{itemize}
	\item $\IS_1 = \{x_{i, k} x_{j, k} \imp z \mid i \in [m], k \in [3] \}$,
	\item $\IS_2 = \{y_i \imp z \mid i \in [m] \}$,
	\item $\IS_3 = \{x_{i, k} z \imp y_i \mid i \in [m], k \in [3] \}$.
\end{itemize}
And let $\IS_Y = \IS_1 \cup \IS_2 \cup \IS_3$. 
%Finally, we add the unique key $K = \{y_1, \dots, y_m, z\}$ and the set $\M = \{Y 
%\setminus \{x_{i, 1}, x_{i, 2}, x_{i, 3}, y_i\} \mid i \in [m]\}$.
In \cite{kavvadias2000generating} the authors show that \csmc{CE} is 
already intractable for these instances.

Therefore, applying the reduction from Lemma \ref{lem:reduction}, we obtain that 
\csmc{MCCEnum} cannot be solved in output-polynomial time in the following case: $\U = Y 
\cup \{u, v\}$, $\IS = \IS_Y \cup \{Y \imp uv \}$, $\graph = (\U, \edges = \{uv\})$.

Let us show that $\cs$, the closure system associated to $\IS$, is indeed lower bounded.
We proceed by analysing the $D$-relation.
Observe first that $\cs$ is standard.
We begin with $u, v$.
Let $t \in \U \setminus \{u\}$ and $M \in \M(\cs)$ such that $t \detm M$.
As no premise of $\IS$ contains $u$, it must be that $u \in M$.
Hence for any $t \in \U \setminus \{u\}$, $t$ does not depend on $u$.
Applying the same reasoning on $v$, we obtain that no $D$-cycle can contain $u$ or $v$.
Let $x_i \in \U$, $i \in [n]$.
As $x_i$ is the conclusion of no implication in $\IS$, we have that the unique 
meet-irreducible element $M_i$ satisfying $x_i \detm M_i$ is $\U \setminus x_i$.
Therefore, there is no element in $\U \setminus \{x_i\}$ on which $x_i$ depends, so that 
no $D$-cycle can contain $x_i$, for any $i \in [n]$.
Let us move to $z$.
As $y_j \imp z \in \IS$ for any $j \in [m]$, we have ${y_j}_* = \phi(y_j)_* = \{z\}$.
Hence,  $zDy_j$ cannot hold since $M \jdet y_j$ implies $z \in M$, for any $M \in 
\M(\cs)$.
Thus, $z$ only depends on some of the $x_i$'s, $i \in [n]$, and no $D$-cycle can 
contain $z$ either.

Henceforth, the only possible $D$-cycles must be contained in $\{y_1, \dots, y_m\}$.
We show that for any $i, k \in [m]$, $y_i D y_k$ does not hold.
For any $y_i$, $i \in [m]$, we have ${y_i}_* = \{z\}$ as $y_i \imp z \in \IS$.
Hence, a meet-irreducible element $M_i$ satisfying $y_i \detm M_i \jdet y_k$ must contain 
$z$.
Let $F \in \cs$ be any closed set satisfying $y_i \notin F$ but $z \in F$.
Assume there exists some $y_k$ such that $y_k \notin F$.
Then $F \cup \{y_k\} \in \cs$, as $y_k \imp z$ is the only implication having $y_k$ in 
its premise, and $z \in F$.
Therefore, it must be that for any $M_i \in \M(\cs)$ such that $z \in M_i$ and 
$y_i \notin M_i$, $\{y_1, \dots, y_m \} \setminus \{y_i\} \subseteq M_i$ is verified, so 
that $y_i \detm M_i \jdet y_k$ is not possible.
As a consequence $y_i D y_k$  cannot hold, for any $i, k \in [m]$.
We conclude that $\cs$ has no $D$-cycles and that it is lower bounded by Theorem
\ref{cor:freese-LB}.
\qed
\end{proof}

Therefore, there is no algorithm solving \csmc{MCCEnum} in output-polynomial time 
unless $\Po = \NP$ even when restricted to lower bounded closure systems.
In the next section, we consider classes of closure systems where \csmc{MCCEnum} can be 
solved in incremental-polynomial time.

\section{Minimal generators with bounded size}
\label{sec:mingen}

Let $\IS$ be an implicational base over $\U$ and $\graph$ a non-empty consistency-graph.
Observe that $\is(\graph) \cup \{\U\}$ is a closure system where a set $F \subseteq \U$ 
is closed if and only if $F = \U$ or it is an independent set of $\graph$.
From this point of view, elements of $\Sol$, are those maximal proper subsets of $\U$ 
that are both closed in $\cs$ and $\is(\graph) \cup \{\U\}$.
Consequently, the maximal consistent closed sets of $\cs$ with respect to $\graph$ are 
exactly the co-atoms of $\cs \cap (\is(\graph) \cup \{\U\})$.
Now, if we can guarantee that $\K$, the keys of $\cs \cap (\is(\graph) \cup 
\{\U\})$, has polynomial size with respect to $\IS$, $\U$ and $\graph$, we can derive an 
incremental-polynomial time algorithm computing $\Sol$ in two steps: 
\begin{enumerate}
	\item Compute the set of keys $\K$ which has polynomial size with respect to $\U$,
	\item Compute $\mis(\K) = \Sol$.
\end{enumerate}

To identify cases where $\K$ has polynomial size with respect to $\IS, \U$ and $\graph$, 
the first step is to characterize its elements.
To do so, we have to guarantee that a set $Y \subset \U$ contains a key of $\K$ 
whenever $Y$ or $\cl(Y)$ is inconsistent with respect to $\graph$.
Looking at $\graph$ is sufficient to distinguish between consistent and inconsistent 
closed sets of $\cs$.
However, there may be consistent (non-closed) sets $Y$ such that $\cl(Y)$ contains 
an edge of $\graph$.
These will not be seen by just considering $\graph$.
Thus, if $uv$ is the edge of $\graph$ contained in $\cl(Y)$, we deduce that there must be 
a minimal generator $A_u$ of $u$ contained in $Y$, possibly $A_u = \{u\}$.
Similarly, $Y$ contains a minimal generator $A_v$ of $v$.
In particular, keys in $\K$ will share the following property.

\begin{proposition} \label{prop:key-union}
Let $K \in \K$.
Then there exists $uv \in \edges$, a minimal generator $A_u$ of $u$, and a minimal 
generator $A_v$ of $v$ such that $K = A_u \cup A_v$.
\end{proposition}

\begin{proof}
Let $K \in \K$.
By assumption, $\cl(K)$ contains an edge $uv$ of $\graph$.
Thus, there exists minimal generators $A_u$ of $u$ and $A_v$ of $v$ such that $A_u \cup 
A_v \subseteq K$.
Assume that $A_u \cup A_v \subset K$ and let $x \in K \setminus (A_u \cup A_v)$.
As $u \in \cl(A_u)$ and $v \in \cl(A_v)$, we get $uv \in \cl(K \setminus 
\{x\})$, a contradiction with the minimality of $K$.
\qed
\end{proof}

\begin{example}
We consider $\IS$, $\U$ and $\graph$ of Figure \ref{fig:running-ex}.
We have that $\cl(135) = 1235$ is inconsistent as it contains $25$.
However $135$ is consistent with respect to $\graph$.
For this example, we will have $\K = \{135, 34, 24, 25\}$. 
Note that $135$ can be decomposed following Proposition \ref{prop:key-union} as the 
minimal generator $13$ of $2$, and $5$ as a trivial minimal generator for itself.
\end{example}
%

% je ne pense pas modifier
Remark that  $\edges \nsubseteq \K$ in the general case, as there may be an 
implication $u \imp v$ in $\IS$ for some inconsistent pair $uv \in \edges$.
Thus $u$ is a key which satisfies Proposition \ref{prop:key-union} with $A_u = A_v = 
\{u\}$.
It also follows from Proposition \ref{prop:key-union} that $\C(\cs)$ plays an important 
role for \csmc{MCCEnum}.
When no restriction on $\C(\cs)$ holds, $\K$ can have exponential 
size with respect to $\IS$ and $\graph$.
The next example drawn from \cite{kavvadias2000generating} illustrates this exponential 
growth.

\begin{example} \label{ex:ranked-exponential}
Let $\U = \{x_1, \dots, x_n, y_1, \dots, y_n, u, v\}$ and $\IS = \{x_i \imp y_i \mid i 
\in [n]\} \cup \{ y_1 \dots y_n \imp uv \}$.
The consistency-graph is $\graph = (\U, \{uv\})$.
The set of non-trivial minimal generators of $u$ and $v$ is $\{z_1 \dots z_n \mid z_i \in 
\{x_i, y_i\}, i \in [n] \}$.
Moreover, minimal generators of $u$ and $v$ are also the keys of $\cs \cap (\is(\graph) 
\cup \{\U\})$.
Thus, $\card{\K} = 2^n$, which is exponential with respect to $\IS$ and $G$.
Observe that $\IS$ is acyclic \cite{hammer1995quasi, wild2017joy}: for any $x, y \in \U$ 
if $y$ belongs to some minimal generator of $x$, then $x$ is never contained in a minimal 
generator of $y$.
\end{example}

Hence, computing $\Sol$ through the intermediary of $\K$ is in general 
impossible in output-polynomial time.
In fact, this exponential blow up occurs even for small classes of closure 
systems where the Carath\'{e}odory number $\C(\cs)$ is unbounded.
In Example \ref{ex:ranked-exponential} for instance, the closure system induced by $\IS$ 
is acyclic \cite{hammer1995quasi, wild2017joy}, a particular case of lower boundedness 
\cite{adaricheva2017optimum}.

On the other hand, let us assume now that $\C(\cs)$ is bounded by some constant $k \in 
\cb{N}$.
By Proposition \ref{prop:key-union}, every key in $\K$ has at most $2 \times k$ elements.
As a consequence we show in the next theorem that the two-steps algorithm we described 
can be conducted in incremental-polynomial time.

\begin{theorem} \label{thm:incr-poly-k-mingen}
Let $\IS$ be an implicational base over $\U$ with induced $\cs$, and $\graph = (\U, 
\edges)$ a consistency-graph.
If $\C(\cs) \leq k$ for some constant $k \in \cb{N}$,
the problem \csmc{MCCEnum} can be solved in incremental-polynomial time.
\end{theorem}

\begin{proof}
The set of keys $\K$ can be computed in incremental-polynomial time with respect to $\K$, 
$\IS$, $\U$ and $\graph$ using the algorithm of Lucchesi and Osborn 
\cite{lucchesi1978candidate} with input $\IS' = \IS \cup \{uv \imp \U \mid uv \in 
\edges\}$.
Observe that the closure system associated to $\IS'$ is exactly $\cs \cap \{\is(\graph) 
\cup \{\U\}\}$.
Indeed, a consistent closed set of $\cs$ models $\IS'$ and a subset $F \subseteq \U$ 
which satisfies $\IS'$ must also satisfy $\IS$ and being an independent set of $\graph$ 
if $F \subset \U$.
Note that $\K$ is then computed in time $\csf{poly}(\card{\IS} + \card{\U} + 
\card{\graph} + \card{\K})$.
As the total size of $\K$ is bounded by $\card{\U}^{2k}$ by Proposition 
\ref{prop:key-union}, we get that $\K$ is computed in time $\csf{poly}(\card{\IS} + 
\card{\U} + \card{\graph})$.

Then, we apply the algorithm of Eiter and Gottlob \cite{eiter1995identifying} to compute 
$\mis(\K) = \Sol$ which runs in incremental polynomial time.
Since $\K$ has polynomial size with respect to $\card{\U}$, the delay between the $i$-th 
and $(i+1)$-th solution of $\Sol$ output is bounded by $\poly(\card{\U}^{2k} + i)$, that 
is  $\poly(\card{\U} + i)$.
Furthermore, the delay after the last output is also bounded by $\poly(\card{\U}^{2k}) = 
\poly(\card{\U})$.
As the time spent before the first solution output is $\csf{poly}(\card{\IS} + \card{\U} 
+ \card{\graph})$, the whole algorithm has incremental delay as expected.
\qed
\end{proof}

To conclude this section, we show that Theorem \ref{thm:incr-poly-k-mingen} applies to 
various classes of closure systems present in the theory of convex geometries 
\cite{korte2012greedoids}.

A closure system $\cs$ is \emph{distributive} if for any $F_1, F_2 \in \cs$, $F_1 \cup 
F_2 \in \cs$.
Implicational bases of distributive closure systems have premises of size one 
\cite{gratzer2011lattice}.

Let $P = (\U, \leq)$ be a partially ordered set, or poset.
A subset $Y \subseteq \U$ is \emph{convex} in $P$ if for any triple $x \leq y \leq z$, 
$x, z \in Y$ implies $y \in Y$.
The family $\{Y \subseteq \U \mid Y \textrm{ is convex in } P \}$ is known to 
be closure system over $\U$ \cite{birkhoff1985convexity, 
kashiwabara2010characterizations}.

Let $G = (\U, E)$ be a graph.
We say that $G$ is \emph{chordal} if every it has no induced cycle of size $\geq 4$.
A \emph{chord} in a path from $x$ to $y$ is an edge connecting to non-adjacent vertices 
of the path. 
A subset $Y$ of $\U$ is \emph{monophonically convex} in $G$ if for every 
pair $x, y$ of elements in $Y$, every $z \in \U$ which lies on a 
chordless path from $x$ to $y$ is in $Y$.
The family $\{Y \subseteq \U \mid Y \textrm{ is monophonically convex in } G \}$ is a 
closure system \cite{farber1986convexity, korte2012greedoids}.

Finally, let $\U \subseteq \cb{R}^n$, $n \in \cb{N}$, be a finite set of points, and 
denote by $\csf{ch}(Y)$ the \emph{convex hull} of $Y$.
The set system $\{\csf{ch}(Y) \mid Y \subseteq \U\}$ forms a closure system 
\cite{korte2012greedoids} usually known as an \emph{affine convex geometry}.

\begin{corollary}
Let $\IS$ be an implicational base over $\U$ and $\graph = (\U, \edges)$.
\csmc{MCCEnum} can be solved in incremental-polynomial time in the following cases:
\begin{itemize}
	\item $\cs$ is distributive,
	\item $\cs$ is the family of convex subsets of a poset,
	\item $\cs$ is the family of monophonically convex subsets of a chordal graph,
	\item $\cs$ is an affine convex geometry in $\cb{R}^k$ for a fixed constant $k$.
\end{itemize}
\end{corollary}

\begin{proof}
Distributive lattices have Carath\'{e}odory number $1$ as they can be represented by 
implicational bases with singleton premises.
The family of convex subsets of a poset has Carath\'{e}odory number 2 
\cite{kashiwabara2010characterizations} (Corollary 13).
The family of monophonically convex subsets of a chordal graphs have Carath\'{e}odory 
number at most 2 \cite{farber1986convexity} (Corollary 3.4).
The Carath\'{e}odory number of an affine convex geometry in $\cb{R}^k$ is $k - 1$ 
(see for instance \cite{korte2012greedoids}, p. 32). % citer Caratheodory theorem plutot
\qed
\end{proof}

In the distributive case, the algorithm can perform in polynomial delay using 
the algorithm of \cite{johnson1988generating} since $\K$ will be a graph by Proposition 
\ref{prop:key-union}.
This connects with previous results on distributive closure systems by Kavvadias et al. 
\cite{kavvadias2000generating}.

\section{Biatomic atomistic closure systems} 
\label{sec:biatomic}

In this section, we are interested in biatomic atomistic closure systems.
Namely, we show that when minimal generators obey an independence condition, the size of 
$\U$ is exponential with respect to $\C(\cs)$. 
To do so, we show that in biatomic atomistic closure systems, each subset of a minimal 
generator is itself a minimal generator.
This result applies to atomistic modular closure systems, which can be represented by 
implications with premises of size at most $2$ \cite{wild2000optimal}.
This suggests that \csmc{MCCEnum} becomes more difficult when implications have binary 
premises.

First, we need to define atomistic biatomic closure systems.
let $\cs$ be a closure system over $\U$ with associated closure operator $\cl$.
We say that $\cs$ is \emph{atomistic} if for any $x \in \U$, $\cl(x) = \{x\}$.
Equivalently, $\cs$ is atomistic if its join-irreducible elements equal its atoms.
Note that in a standard closure system, an atom is a singleton element.
Biatomic closure systems have been studied by Birkhoff and Bennett 
in \cite{bennett1987biatomic, birkhoff1985convexity}.
We reformulate their definition in terms of closure systems.
A closure system $\cs$ is \emph{biatomic} if for every closed sets $F_1, F_2 \in 
\cs$ and any atom $\{x\} \in \cs$, $x \in \cl(F_1 \cup F_2)$ implies the 
existence of atoms $\{x_1\} \subseteq F_1$, $\{x_2\} \subseteq F_2$ such that $x \in 
\cl(x_1x_2)$.
In atomistic closure systems in particular, the biatomic condition applies to every 
element of $\U$.
Hence the next property of biatomic atomistic closure systems.

\begin{proposition} \label{prop:biatomicity}
Let $\cs$ be a biatomic atomistic closure system.
Let $F \in \cs$ and  $x, y \in \U$ with $x, y \notin F$.
If $y \in \cl(F \cup \{x\})$, then there exists an element $z \in F$ such that $y \in 
\cl(xz)$.
\end{proposition}

\begin{proof}
In atomistic closure systems, every element of $\U$ is closed, therefore we apply the 
definition to the closed sets $F$ and $\{x\}$.
\qed
\end{proof}

We will also make use of the following folklore result about minimal generators.
We give a proof for self-containment.

\begin{proposition} \label{prop:trace}
If $A_x$ is a minimal generator of $x \in \U$, then $\cl(A) \cap A_x 
= A$ for any $A \subseteq A_x$. 
\end{proposition}

\begin{proof}
First, we have that $A \subseteq \cl(A) \cap A_x$ as $A \subseteq \cl(A)$ and $A 
\subseteq A_x$.
Now suppose that there exists $a \in \cl(A) \cap A_x$ such that $a \notin A$.
Then, $a \in \cl(A_x \setminus \{a\})$ as $A \subseteq A_x \setminus \{a\}$.
Hence, $\cl(A_x) = \cl(A_x \setminus \{a\})$ and $x \in \cl(A_x \setminus \{a\})$, a 
contradiction with $A_x$ being a minimal generator of $x$.
\qed
\end{proof}

Our first step is to show that in a biatomic atomistic closure system, if $A_x$ is a 
minimal generator for some $x \in \U$, then every non-empty subset $A$ of $A_x$ is itself 
a minimal generator for some $y \in \U$. 
We prove this statement in Lemmas \ref{lem:induction} and \ref{lem:mingen}.
Recall that an element $x \in \U$ is a (trivial) minimal generator of itself.

\begin{lemma} \label{lem:induction}
%Let $\cs$ be a biatomic atomistic closure system.
Let $x \in \U$ and let $A_x$ be a minimal generator of $x$ with size $k \geq 2$.
Then for any $a_i \in A_x$, $i \in [k]$, there exists $y_i \in \U$ such that $A_x 
\setminus \{a_i\}$ is a minimal generator of $y_i$.
\end{lemma}

\begin{proof}
Let $A_x = \{a_1, \dots, a_k\}$ be a minimal generator of $x$ such that $k \geq 2$.
Then, for any $a_i \in A_x$, $i \in [k]$, we have $a_i \notin \cl(A_x \setminus \{a_i\})$ 
by Proposition \ref{prop:trace}.
However, we have $x \in \cl(\{a_i\} \cup \cl(A_x \setminus \{a_i\})) = \cl(A_x)$.
Thus, by Proposition \ref{prop:biatomicity}, there must exists $y_i \in 
\cl(A_x \setminus \{a_i\})$ such that $x \in \cl(a_iy_i)$.
%The situation is illustrated in Figure \ref{fig:lemma}.

Let us show that $A_x \setminus \{a_i\}$ is a minimal generator of $y_i$.
Assume for contradiction this is not the case.
As $y_i \in \cl(A_x \setminus \{a_i\})$, there must be a proper subset $A$ of $A_x 
\setminus \{a_i\}$ which is a minimal generator for $y_i$.
Note that since $A_x$ has at least $2$ elements, at least one proper subset of $A_x 
\setminus \{a_i\}$ exists.
As $A \subset A_x \setminus \{a_i\}$, there exists $a_j \in A_x$, $a_j \neq a_i$, such 
that $a_j \notin A$.
Therefore, $A \subseteq A_x \setminus \{a_j\}$ and $\cl(A) \subseteq \cl(A_x \setminus 
\{a_j\})$.
More precisely, $y_i \in \cl(A)$ and hence $y_i \in \cl(A_x \setminus \{a_j\})$.
However, we also have that $a_i \in \cl(A_x \setminus \{a_j\})$, and since $x \in 
\cl(a_iy_i)$, we must have $x \in \cl(A_x \setminus \{a_j\})$, a contradiction with $A_x$ 
being a minimal generator of $x$.
Thus, we deduce that $A_x \setminus \{a_i\}$ is a minimal generator for $y_i$, concluding 
the proof.
\qed
%In particular, $y_i \notin A_x$.
%Indeed, assume for contradiction $y_i \in A_x$.
%Then, $a_iy_i \subset A_x$ since $A_x$ has size $k \geq 3$, and $x \in \cl(a_iy_i)$.
%However, this contradicts $A_x$ being a minimal generator of $x$ with size $k \geq 3$.
%Thus, we must have $y_i \notin A_x$ for any $i \in [k]$.
%Now, for any $i \neq j$, $j \in [k]$, we have $a_i \in A_x \setminus \{a_j\}$ and hence 
%$a_i \in \cl(A_x \setminus \{a_j\})$.
%Since $x \in \cl(a_iy_i)$, $a_i \in \cl(A_x \setminus \{a_j\})$, and $A_x$ is a minimal 
%generator of $x$, it follows that $y_i \notin \cl(A_x \setminus \{a_j\})$.
%
%We move to the second item and show that $A_x \setminus \{a_i\}$ is a minimal generator 
%of $y_i$.
%Let $A$ be a proper subset of $A_x \setminus \{a_i\}$, that is $A \subset A_x \setminus 
%\{a_i\}$.
%As $A_x$ has at least $3$ elements, such a $A$ exists.
%Then, there exists $a_j \in A_x$ such that $a_j \notin A$.
%Therefore, we have that $A \subseteq A_x \setminus \{a_j\}$ and hence that $\cl(A) 
%\subseteq \cl(A_x \setminus \{a_j\})$ since $\cl$ is isotone.
%As we have shown that $y_i \notin \cl(A_x \setminus \{a_j\})$, it follows that $y_i 
%\notin \cl(A)$.
%Thus, for any proper subset $A$ of $A_x \setminus \{a_i\}$, $y_i \notin \cl(A)$, while 
%$y_i \in \cl(A_x \setminus \{a_i\})$.
%Consequently, $A_x \setminus \{a_i\}$ is indeed a minimal generator for $y_i$, 
%concluding 
%the proof.
%\qed
\end{proof}

%\begin{figure}[ht!]
%\centering 
%\includegraphics[scale=1.0]{Figures/lemma2.pdf}
%\caption{The situation of Lemma \ref{lem:induction}.}	
%\label{fig:lemma}
%\end{figure}

In the particular case where $A_x$ has only two elements, say $a_1$ and $a_2$, then $A_x 
\setminus \{a_1\} = \{a_2\}$ and the element $a_2$ is a trivial minimal generator of 
itself.
By using inductively Lemma \ref{lem:induction} on the size of $A_x$, one can derive the 
next straightforward lemma.

\begin{lemma} \label{lem:mingen}
Let $\cs$ be a biatomic atomistic closure system.
Let $A_x$ be a minimal generator of some $x \in \U$.
Then, for any $A \subseteq A_x$ with $A \neq \emptyset$, there exists $y \in \U$ such 
that $A$ is a minimal generator of $y$.
\end{lemma}

Thus, for a given minimal generator $A_x$ of $x$, any non-empty subset $A$ of $A_x$ is 
associated to some $y \in \U$.
We show next than when $A_x$ also satisfies an independence condition, $A$ will be the 
unique subset of $A_x$ associated to $y$.
Following \cite{gratzer2011lattice}, we reformulate the definition of independence in an 
atomistic closure system $\cs$, but restricted to its atoms.
A subset $Y$ of $\U$ is \emph{independent} in $\cs$ if for any $Y_1, Y_2 \subseteq Y$,
$\cl(Y_1 \cap Y_2) = \cl(Y_1) \cap \cl(Y_2)$. 

\begin{lemma} \label{lem:unique}
Let $\cs$ be a biatomic atomistic closure system.
Let $A_x$ be an independent minimal generator of $x \in \U$, and let $A$ be a non-empty 
subset of $A_x$.
Then, there exists $y \in \U$ such that $A$ is the unique minimum subset of $A_x$ 
satisfying $y \in \cl(A)$.
\end{lemma}

\begin{proof}
Let $A_x$ be an independent minimal generator of $x \in \U$, and let $A$ be a non-empty 
subset of $A_x$.
By Lemma \ref{lem:mingen}, there exists $y \in \U$ such that $A$ is a minimal generator 
for $y$, which implies $y \in \cl(A)$.

To prove that $A$ is the unique minimum subset of $A_x$ such that $y \in \cl(A)$, 
we show that for any $B \subseteq A_x$ such that $A \nsubseteq B$, $y \in \cl(B)$ cannot 
hold.
Consider $B \subseteq A_x$ with $A \nsubseteq B$ and suppose that $y \in \cl(B)$.
Note that $B$ must exist as the empty set is always a possible choice.
Since $y \in \cl(A)$, we have $y \in \cl(A) \cap \cl(B)$.
Furthermore, $\cl(A \cap B) \subset \cl(A)$ as $A \cap B \subset A$ and $\cl(A \cap B) 
\cap A_x = A \cap B$ by Proposition \ref{prop:trace}.
Moreover, $A_x$ is independent, so that $\cl(A) \cap \cl(B) = \cl(A \cap B)$.
Hence, $y \in \cl(A \cap B) \subset \cl(A)$, a contradiction with $A$ being 
a minimal generator of $y$.
\qed
\end{proof}

Thus, when $A_x$ is independent, each non-empty subset $A$ of $A_x$ is the unique minimal 
generator of some $y$ being included in $A_x$.
As a consequence, we obtain the following theorem.

\begin{theorem} \label{thm:bounded-exp}
Let $\cs$ be a biatomic atomistic closure system. 
If for any $x \in \U$ and any minimal generator $A_x$ of $x$, $A_x$ is independent,
then $\C(\cs) \leq \lceil \log_2(\card{\U} + 1) \rceil$.
\end{theorem}

\begin{proof}
Let $A_x$ be a minimal generator of $x$, $x \in \U$ such that $\C(\cs) = 
\card{A_x}$.
As $A_x$ is a minimal generator, $\cl(A) \neq \cl(A')$ for any distinct $A, A' \subseteq 
A_x$, due to Proposition \ref{prop:trace}. 
Furthermore $A_x$ is independent by assumption.
Thus, by Lemma \ref{lem:unique}, for each non-empty subset of $A$, there exists $y \in 
\U$ such that $A$ is the unique minimum subset of $A_x$ with $y \in \cl(A)$.
Consequently, $\U$ must contain at least $2^{\card{A_x}} - 1$ elements in order to cover 
each non-empty subset of $A_x$, that is $2^{\card{A_x}} - 1 \leq \card{\U}$, which can be 
rewritten as $\card{A_x} = \C(\cs) \leq \lceil \log_2(\card{\U} + 1) \rceil$ as required.
\qed
\end{proof}

Now let $\cs$ be a biatomic atomistic closure system on $\U$ given by some implicational 
base $\IS$ and let $\graph = (\U, \edges)$ be a consistency-graph.
Assume that every minimal generator is independent.
By Theorem \ref{thm:bounded-exp}, we have that $\card{\U}$ has exponential size with 
respect to $\C(\cs)$, and by Proposition \ref{prop:key-union}, it must be that 
the size of a key in $\K$ cannot exceed $2 \times \lceil \log_2(\card{\U} + 1) \rceil$.
Thus, with respect to $\IS$, $\graph$ and $\U$, $\K$ will have size quasi-polynomial in 
the worst case.
Using the same algorithm as in Section \ref{sec:mingen}, we obtain the next theorem.

\begin{theorem} \label{thm:quasi-poly}
Let $\IS$ be an implicational base of a biatomic atomistic closure system $\cs$ over $\U$ 
and $\graph$ a consistency-graph.
If for any $x \in \U$ and any minimal generator $A_x$ of $x$, $A_x$ is independent, then 
\csmc{MCCEnum} can be solved in output-quasipolynomial time.
\end{theorem}

\begin{proof}
For clarity, we put $n = \card{\U}$ and $k$ as the total size of the output 
$\mis(\K)$.
%Note that $k$ is bounded by $\card{\U} \times \card{\mis(\K)}$.
$\K$ can be computed in incremental-polynomial time with the algorithm in 
\cite{lucchesi1978candidate}.
Furthermore, by Theorem \ref{thm:bounded-exp}, the total size of $\K$ is bounded by 
$n^{\log(n)}$.
Thus, this first step runs in time $\csf{poly}(\card{\IS} + \card{\graph} + n + 
n^{\log(n)})$, which is bounded by $\csf{poly}(\card{\IS} + \card{\graph} + 
n)^{\log(n)}$ being quasipolynomial in the size of $\IS$, $\graph$, $\K$ and $\U$.
To compute $\mis(\K) = \Sol$ we use the algorithm of Fredman and Khachiyan 
\cite{fredman1996complexity} whose running time is bounded by $(n^{\log(n)} + 
k)^{o(\log(n^{\log(n)} + k))}$.
In our case, we can derive the following upper bounds:
\begin{align*}
(n^{\log(n)} + k)^{o(\log(n^{\log(n)} + k))}
& \leq (k + n)^{\log(n) \times o(\log(k + n)^{\log(n)})} \\ 
& \leq (k + n)^{O(\log^3(k + n))}
\end{align*}
Thus, the time needed to compute $\mis(\K)$ from $\K$ is output-quasipolynomial in 
the size of $\U$ and $\Sol$.
Consequently, the running time of the whole algorithm is bounded by 
\[\csf{poly}(\card{\IS} + \card{\graph} + n)^{\log(n)} + (k + n)^{O(\log^3(k + n))} \] 
which is indeed quasipolynomial in the size of the input $\IS$, $\U$, $\graph$ and the 
output $\mis(\K) = \Sol$.
\qed
\end{proof}

To conclude this section, we show that atomistic modular closure systems 
\cite{stern1999semimodular, gratzer2011lattice} satisfy conditions of Theorem 
\ref{thm:quasi-poly}. 
Recall that a closure system $\cs$ is modular if for any $F_1, F_2, F_3 \in \cs$, $F_1 
\subseteq F_2$ implies $\cl(F_1 \cup (F_2 \cap F_3)) = \cl(F_1 \cup F_3) \cap F_2$.
It was proved for instance in \cite{bennett1987biatomic} (Theorem 7) that atomistic 
modular closure systems are biatomic.
To show that any minimal generator is independent, we make use of the following result.

\begin{theorem}(Reformulated from \cite{gratzer2011lattice}, Theorem 360) 
\label{thm:gratzer}
Let $\cs$ be a modular closure system. 
A subset $A = \{a_1, \dots, a_k\} \subseteq \U$ is independent if and only if
$ \cl(a_1) \cap \cl(a_2) = \cl(a_1 a_2) \cap \cl(a_3) = \dots = \cl(a_1 \dots a_{k - 1}) 
\cap \cl(a_k) = \emptyset$.

\end{theorem}

\begin{proposition} \label{prop:independence-modular}
Let $\cs$ be an atomistic modular closure system.
Let $A_x$ be a minimal generator of some $x \in \U$. 
Then $A_x$ is independent.
\end{proposition}

\begin{proof}
Let $A_x = \{a_1, \dots, a_k\}$ be a minimal generator for some $x \in \U$.
Then, by Proposition \ref{prop:trace}, $\cl(a_1 \dots a_i) \cap A_x = a_1 \dots 
a_i$ for any $i \in [k]$.
Furthermore, $\cl(a) = \{a\}$ for any $a \in \U$ since $\cs$ is atomistic.
Thus we conclude that $\cl(a_1 \dots a_i) \cap \cl(a_{i + 1}) = \emptyset$ for any $i\in 
[k - 1]$ as $a_{i + 1} \notin a_1 \dots a_i$.
It follows by Theorem \ref{thm:gratzer} that $A_x$ is indeed independent.
\qed
\end{proof}

\begin{corollary}
Let $\IS$ be an implicational base over $\U$ and $\graph = (\U, \edges)$.
Then \csmc{MCCEnum} can be solved in output-quasipolynomial time if:
\begin{itemize}
	\item $\cs$ is biatomic atomistic and has Carath\'{e}odory number $2$ (including 
	convex subsets of a poset and monophonically convex sets of a chordal graph),
	\item $\cs$ is atomistic modular.
\end{itemize}
\end{corollary}

\begin{proof}
For the first statement, note that in an atomistic closure system with Carath\'{e}odory 
number $2$, any minimal generator $A_x$ contains exactly two elements $a_1, a_2$.
Since $\cs$ is atomistic, $a_1$ and $a_2$ are closed and the independence of $A_x$ 
follows.

If $\cs$ is atomistic modular, biatomicity follows from \cite{bennett1987biatomic} 
(Theorem 7), and independence from Proposition \ref{prop:independence-modular}.
\qed
\end{proof}

\begin{remark}
For atomistic modular closure systems, the connection between the size of $\U$ and the 
Carath\'{e}odory number may also be derived from counting arguments on subspaces of 
vector spaces \cite{wild1996minimal}.
\end{remark}

\section{Conclusion} 
\label{sec:conclusion}

In this paper we proved that given a consistency-graph over an implicational base, the 
enumeration of maximal consistent closed sets is impossible in output-polynomial time 
unless $\Po = \NP$.
Moreover, we showed that this problem, called \csmc{MCCEnum}, is already intractable for 
the well-known class of lower bounded closure systems.
On the positive side, we proved that when the size of a minimal generator is bounded by a 
constant, the enumeration of maximal consistent closed sets can be conducted in 
incremental polynomial time.
This result covers various classes of convex geometries.
Finally, we proved that in biatomic atomistic closure systems, \csmc{MCCEnum} can be 
solved in output-quasipolynomial time provided minimal generators obey an independence 
condition.
This applies in particular to atomistic modular closure systems.
In Figure \ref{fig:hierarchy}, we summarize our results in the hierarchy of closure 
systems.

For future research, we would like to understand which properties or paramaters of 
closure systems make the problem intractable or solvable in output-polynomial time. 
For instance, we have seen that a bounded Carath\'{e}odory number gives an 
incremental-polynomial time algorithm, while lower boundedness makes the problem 
intractable.
Another question is the following: is the problem still hard if the closure system is 
given by a context (equivalently, its meet-irreducible elements)?  
The question is particularly interesting for classes such as semidistributive 
closure systems where we can compute the context in polynomial 
time in the size of an implicational base.

\begin{figure}[ht!]
\includegraphics[scale=0.95, page=3]{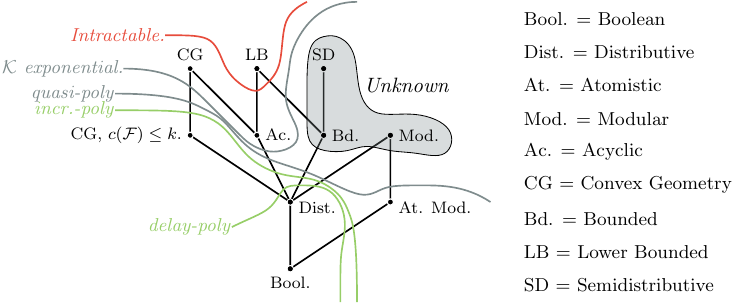}
\caption{The complexity of \csmc{MCCEnum} in the hierarchy of closure systems}
\label{fig:hierarchy}
\end{figure}

\vspace{-2em} 

\bibliographystyle{acm}
\bibliography{ICFCA}

\end{document}